\documentclass[12pt,draftcls,onecolumn]{IEEEtran}
\IEEEoverridecommandlockouts
\usepackage{graphicx}

\usepackage{epstopdf}
\usepackage{xcolor}
\usepackage{amsmath}
\usepackage{amssymb}
\usepackage{amsfonts}
\usepackage{subfigure}
\usepackage{url}
\usepackage{filecontents}
\usepackage{float}
\title{\LARGE \bf Decentralized Charging of Plug-In Electric Vehicles with Distribution Feeder Overload Control}
\author{Abouzar Ghavami, \textit{Member, IEEE}, Koushik Kar, \textit{Member, IEEE}, Aparna Gupta
\thanks{The authors are affiliated with Rensselaer Polytechnic Institute, Troy, NY, USA. Emails: \tt\small ghavamip@gmail.com, \{kark, guptaa\}@rpi.edu.}}
\begin{document}
\maketitle
\let\thefootnote\relax\footnotetext{This work was supported in part by the National Science Foundation through award ECCS-1408333.}
\newtheorem{Def}{Definition}
\newtheorem{Thm}{Theorem}
\newtheorem{Alg}{Algorithm}
\newtheorem{Lem}{Lemma}
\thispagestyle{empty}
\pagestyle{plain}
\begin{abstract}
As the number of charging Plug-in Electric Vehicles (PEVs) increase, due to the limited power capacity of the distribution feeders and the sensitivity of the mid-way distribution transformers to the excessive load, it is crucial to control the amount of power through each specific distribution feeder to avoid system overloads that may lead to breakdowns. In this paper we develop, analyze and evaluate charging algorithms for PEVs with feeder overload constraints in the distribution grid. The algorithms we propose jointly minimize the variance of the aggregate load and prevent overloading of the distribution feeders.
\end{abstract}
\section{Introduction}
\indent In order to reduce dependency on oil, the cost of the fuel consumption, air-pollution in residential areas, and at the same time to increase vehicle engine performance efficiency and use of renewable energy resources for transportation purposes, Plugged-In Electric Vehicles (PEVs) have been getting popular in recent years. The increase in PEV usage would imply a significant increase in the overall load on the electric grid, and therefore efficient management of this excess demand is crucial for the overall efficiency and stability of the grid \cite{HT09}.\\
\indent Based on the US nationwide survey data \cite{US09}, \cite{PR04}, the average US household load in 2009 was $1.3$ KW, while the Level 1 PEV charging uses $1.96$ KW and Level 2 charging uses $7.2$KW of extra power load. The level 1 charging load is about $1.5$ and level 2 charging load is about $5.5$ times the average base household load. If every household owns just one PEV in the near future, the peak demand of the grid load from charging the PEVs can increase the peak load by a factor of $2.5$ to $6.5$ times the current peak load. This peak load would not only increase the peak load that a distribution network (regional or local) draws from the transmission grid, but also causes stress on its feeders and transformers. Therefore, it is necessary to manage the congestion that PEVs charging causes in the feeders to prevent breakdown of transformers and other grid components.\\
\indent However, most of the recent work on coordinated PEV charging have only focused on minimizing the overall variance of the load -- the base (non-PEV) plus PEV power -- in the distribution grid \cite{MCH10}, \cite{GTL13}. This minimizes the (peak-to-average) load that the distribution grid under consideration draws from the transmission grid, but does not necessarily ensure that the transformers and feeders in the distribution network are not overloaded. In fact, as our simulation results in Section~\ref{sec:sim} show, optimal coordinated charging that simply minimizes load variance over time without taking into account feeder overload limits, can seriously overload different components of the distribution grid. It is worth noting that both transformers and distribution feeders have ratings on the load (power, current) they can carry. Transformers are sensitive to temperature and may stop functioning properly beyond a certain temperature threshold: this imposes limits on the power carried by the transformer. Power (current) transmission through short-distance feeders that make up typical distribution networks are constrained by their thermal limits, to avoid power line sag due to high temperature \cite{GE79}.\\
\indent In this paper, therefore, we focus on the goal of minimizing the total load variance in the distribution grid, subject to overload constraints on the distribution feeders that constitute the distribution network. We assume that the overload constraints are specified in the form of a maximum power that a feeder in the distribution grid can carry at any time, which is calculated a priori so as to satisfy the thermal limits of the distribution grid elements. A novel feature of our approaches is the consideration of the distribution network topology in computing the control feedback to the PEVs, which can differ across PEVs, depending on the state of overload in the distribution feeders supplying power to them.  In line with the recent work on this topic \cite{MCH10}, \cite{GTL13},  the algorithms proposed in this paper can be implemented in a decentralized manner. Specifically, we show that the algorithms can be implemented as an iterative price-driven coordination between the utility (or aggregator) and the PEVs (or their smart meters), where the electric utility quotes certain time-dependent, non-linear but single-parameter pricing functions to each PEV, to which the PEVs respond by choosing a charging schedule so as to minimize their individual charging costs (best response).\\
\indent In a recent work, the authors in \cite{HAH12} consider thermal constraints (modeled as a complex function of the transformer load) in considering the PEV charging question. However, the study restricts itself to a single transformer that is at the root of the distribution tree; naturally the network aspects and spatial differences in the control feedback are absent from this approach. Another recent work \cite{ARK13} proposes a distributed charging control algorithm that finds a proportionally fair rate allocation for each PEV with considering the maximum capacity of the transmission feeders; however, whereas \cite{ARK13} does not consider the total amount of energy each PEV gets charged and the optimization is done for only one time slot. 

We consider a convex optimization formulation for the PEV charging problem subject to feeder overload constraints, and present two decentralized approaches for the problem. The first approach (that is based on overload cost-functions and presented in Section \ref{sec:overload-cost}) and its convergence analysis can be viewed as generalization of results in \cite{GTL13}, extended to account for the distribution network topology and feeder overload constraints. We also provide a simpler convergence proof to a more general result by mapping the approach to a gradient projection method. As we consider the feeder capacity constraints, the maximum step size required for convergence of this cost-based method is dependent not only on the number of the PEVs (as in \cite{GTL13}), but also on the maximum depth of the tree topology representing the distribution grid. Furthermore, we also provide efficient algorithms (in Section \ref{sec:comp-algos}) that can be used to compute the charging schedules of individual PEVs once the feedback from the utility (distribution network) is obtained. These two algorithms can be used to solve the individual PEV charging profile selection (optimization) sub-problem of \cite{GTL13} as well. The second decentralized approach that we present (described in Section \ref{sec:primal-dual}) is based on a novel application of a primal-dual method that results in dual based solutions that are amenable to decentralized implementation, despite the non-separability in the problem. Due to its reliance on a primal-dual method instead of a gradient method, this algorithm is not directly comparable with that in \cite{GTL13} in terms of its core technical approach; it also provides different convergence properties from our cost-based approach (and that in \cite{GTL13}). We also interpret (in Section \ref{BestResponse}) the charging profile computation for both proposed overload control approaches as price-driven decentralized best response updates, which has not been done in prior work. Finally, while \cite{GTL13} also presents asynchronous and real-time versions of the algorithm proposed there, we only present and analyze synchronous versions of our algorithms in this work. The convergence results of the cost-based method should naturally extend to asynchronous implementation; the convergence properties of the primal-dual algorithm under asynchronous updates remain open for future investigation.

\section{System Model}\label{sec:model}
\indent We model a system where an electric utility negotiates with $K$ plug-in electric vehicles (indexed $1, ..., K$)  to coordinate their charging schedules, over a distribution network. We discretize the time of day into $T$ units, which are indexed as $1, 2, ..., T$. Let $D(t)$ denote the base demand (aggregated non-PEV demand, assumed to be estimated a priori) over the entire distribution network, and $p_k(t)$ denote the charging power of PEV $k$ at time $t$, for all $t \in \{1, ..., T\}$. Let the vector $\mathbf{p}_k=(p_k(1), ..., p_k(T))$ denote the \emph{charging profile} of PEV $k$ and the vector $\mathbf{p}:= (\mathbf{p}_1, ..., \mathbf{p}_K)$ denote the charging profile of all the $K$ PEVs in the system. Assume that PEV $k$ charges over the time interval $T_k=\{t^s_k, ..., t^f_k\}$, where $t^s_k$ is the charging start time and $t^f_k$ is the charging finish time for charging of PEV $k$, $1 \leq t^s_k < t^f_k \leq T$.

We model the distribution network as a tree rooted at the distribution substation (Figure \ref{Topology}). Each PEV is attached to one of the leaf nodes of the tree (Figure \ref{Topology}). Each PEV $k$ is associated with a set of distribution feeders (links in the tree graph), $\Pi_k$, which transfer power from the distribution substation to PEV $k$, and each distribution feeder $l$ carries power to a set $\Gamma_l$ of PEVs, where $\Gamma_l=\{k: l \in \Pi_k\}$. Let $L$ be the total number of distribution feeders (links) (indexed $1, ..., L$) in the tree topology. Let $K_l=|\Gamma_l|$ denote the number of PEVs that receive electricity through the distribution feeder $l$; obviously $K_l \leq K$.

Let $d_{\max}$ denote the maximum path length (number of distribution feeders) between the distribution substation and any PEV, or in other words, the maximum depth of the tree topology. Let $l_k=|\Pi_k|$ be the number of distribution feeders that transfer power from the distribution substation to PEV $k$; thus $l_k \leq d_{\max}$.\\
\begin{figure}
\centering
\includegraphics[width = 0.5\textwidth]{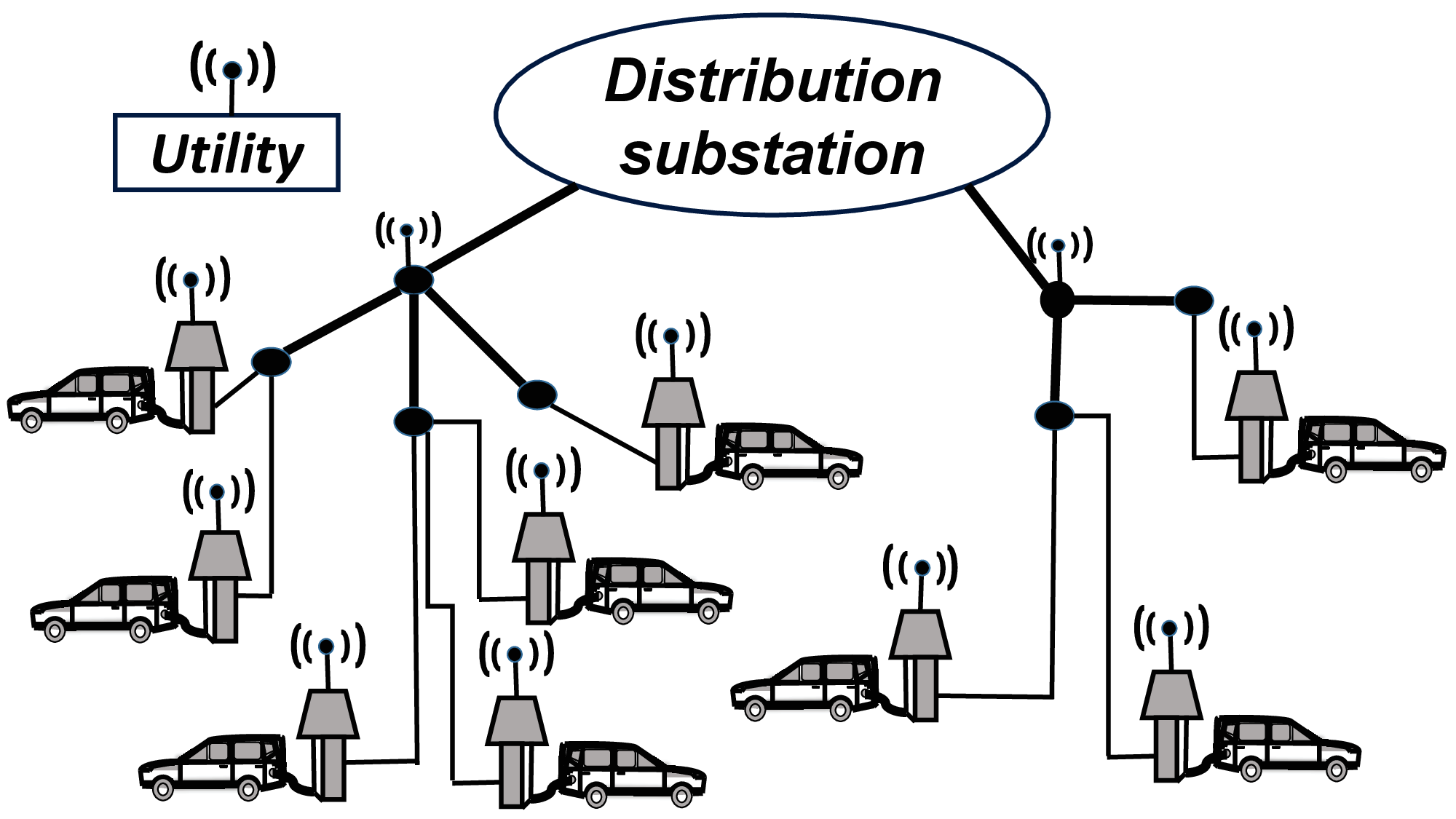}
\caption{Schematic of a distribution grid being used for PEV charging.}
\label{Topology}
\end{figure}
\indent We define $P(t)=\sum_{k=1}^{K}p_k(t)$. Also, $P_l(t)=\sum_{k \in \Gamma_l}p_k(t)$, denoting the total load of PEVs fed from the $l^{th}$ feeder. Let $d_l(t)$ denote the base demand load (non-PEV demand) transmitted through feeder $l$ at time $t$. As each feeder $l$ has a maximum transmission capacity, $\rho_l$, the maximum PEV demand it can support at time $t$ is $P_l^{\max}(t)=\rho_l-d_l(t)$. Thus for each feeder, $l$, we have:
\begin{equation}
P_l(t) \leq \eta_{l,t} P_l^{max}(t),
\label{congestion1}
\end{equation}
where $\eta_{l,t} \leq 1$ is the \textit{overload control} parameter for link $l$ at time $t$. While power flow dynamics can be complex, the simple overload constraints (represented by a power capacity constraint of $\rho_l$ for feeder $l$) are motivated by the operation of feeder/transformer protection systems in real distribution lines, where protection systems monitor the current (or power) flow through the feeder lines/transformers, and simply trip when it exceeds a certain pre-determined threshold \cite{GE79}.\\
\indent For PEV $k$, as per its battery specification, the charging rate should be within a range:
\begin{equation}
0 \leq p_k(t) \leq p_k^{\max}(t), \hspace{1mm} t \in T_k.
\label{limits}
\end{equation}
Obviously, $p_k(t) = 0$ for $t \notin T_k$. Let define $U_k$ as the total charging energy PEV $k$ requests in the interval $T_k$. Hence:
\begin{equation}
\sum_{t\in T_k} p_k(t) = U_k, \hspace{5mm} 1 \leq k \leq K,
\label{Uk}
\end{equation}
where $U_k$ is calculated based on the PEV battery capacity, $B^k$, and charging efficiency, $\nu_k$, as $U_k=\frac{B^k}{\nu_k}$. Let $\mathcal{D}_k=\{\mathbf{p}_k \ | \ 0 \leq p_k(t) \leq p_k^{max}(t), \sum_{t\in T_k}p_k(t)=U_k\}$ denote the set of charging profiles for PEV $k$ satisfying constraints (\ref{limits}) and (\ref{Uk}). Let $\mathcal{D} := \mathcal{D}_1 \times ... \times \mathcal{D}_K$ denote the set of all feasible charging profiles of all PEVs.\\
\indent Subject to constraints (\ref{congestion1}), (\ref{limits}) and (\ref{Uk}), the utility would like to charge the PEVs such that the variance of the total load (PEV plus non-PEV load) in the distribution network is minimized, i.e., it seeks to minimize the following objective:
\begin{equation}
\sum_{t=1}^{T}\left(D(t)+\sum_{k=1}^{K}p_k(t)\right)^2. \label{variance-objective}
\end{equation}
Objective (\ref{variance-objective}) is consistent with that used in prior literature \cite{MCH10}, \cite{GTL13}. Note that this is equivalent to minimizing peak-to-average power ratio in the distribution network, and thereby reduces the load variance in the transmission network (that supplies power to this distribution network through the distribution substation) as well.
Using the notation we have introduced, this can be compactly represented as the \emph{primal problem}, $\mathbf{P}$:
\begin{gather}
\text{$\mathbf{P}$:    } \min_{\mathbf{p} \in \mathcal{D}} f(\mathbf{p}) = \sum_{t=1}^{T}\mathcal{V}\left(D(t)+P(t)\right), \nonumber \\
\text{s.t.: }\hspace{5mm} g_{l,t}(\mathbf{p})=P_l(t) - \eta_{l,t} P_l^{max}(t) \leq 0, \hspace{5mm} \forall l, t,
\label{minim}
\end{gather}
where $\mathcal{V}(x)$ is a strictly convex function whose first and second derivatives are continuous and the second derivative is bounded over $x \in [0, \max_{t}\{\sum_{k=1}^{K} p_k^{\max}(t)\}]$, i.e.  $\exists B_1: \mathcal{V}^{\rq{}\rq{}}(x) \leq B_1$.
In the above formulation, note that the charging constraints of individual PEVs, (\ref{limits}) and (\ref{Uk}), are captured in the constraint set $\mathcal{D}$. Our goal is to obtain decentralized PEV charging solutions that attain the optimal value $f^*$ of this primal problem.
\section{PEV Charging Using Overload Cost Functions} \label{sec:overload-cost}
\indent In the first method, we take into account the overload constraints by associating cost (penalty) functions with the constraints in equation (\ref{congestion1}), and add them to objective (\ref{variance-objective}) to create an \emph{overload-augmented} objective function. More specifically, let each link $l$ be associated with an overload cost function, $C_{l,t}(x)$, at time $t$, where $C_{l,t}(x)$ is a convex, nonnegative function whose first and second derivatives are continuous, and the second derivative is bounded over $x \in \mathbb{R}$, i.e.  $\exists B: C_{l,t}^{\rq{}\rq{}}(x) \leq B$. The primal problem can then be \emph{approximated} by minimizing an overload-discounted objective function $\mathcal{L}(\mathbf{p})$ over the set of all charging profiles $\mathcal{D}$, $\mathbf{P}_1= \min_{\mathbf{p}\in \mathcal{D}}\mathcal{L}(\mathbf{p})$, where
\begin{gather}
\mathcal{L}(\mathbf{p}) = \sum_{t=1}^{T}\left\{\mathcal{V}\left(D(t)+\sum_{k=1}^{K}p_k(t)\right) \right. \nonumber\\
\left. + \sum_{l=1}^{L} C_{l,t}\left(P_l(t)-\eta_{l,t} P_l^{max}(t)\right) \right\}. \label{cost-method-def2}
\end{gather}
Towards developing a distributed algorithm, we solve $\mathbf{P}_1$ using a \textit{gradient projection method} with constant step-size $\alpha$. This method can be viewed as a generalization of the method (Algorithm 1) of \cite{GTL13} such that the feeder overload constraints are also taken into account. Due to the separable nature of the set $\mathcal{D} := \mathcal{D}_1 \times ... \times \mathcal{D}_K$, this can be expressed as the following iterative update procedure,
\begin{equation}
\mathbf{p}_k^{m+1} = \mathcal{P}_{\mathcal{D}_k}[\mathbf{p}_k^{m}-\alpha \mathcal{L}_{\mathbf{p}_k}(\mathbf{p}^m)], \ \ m=0,1, \ldots,
\label{renew1}
\end{equation}
where $\mathcal{P}_{\mathcal{D}_k}[\mathbf{x}]$ is the projection of the vector $\mathbf{x}$ over the hyperplane $\mathcal{D}_k$ and $\mathcal{L}_{\mathbf{p}_k}(\mathbf{p})=[\nabla \mathcal{L}(\mathbf{p})]_{\mathbf{p}_k}$ is the gradient of $\mathcal{L}(\mathbf{p})$ with respect to $\mathbf{p}_k$.
Also, let $p_k^{0}(.):=0$, $\forall k$, and $\alpha < \frac{1}{K(B_1+d_{max}B)}$.
\begin{Thm}
If $\mathcal{D}$ is not empty, then for $0 < \alpha < (K(B_1+d_{max}B))^{-1}$, $\mathcal{L}$ decreases monotonically as $m$ increases and $\mathcal{L}^{*}=\lim_{m \rightarrow \infty}\mathcal{L}(\mathbf{p}^m)$ minimizes $\mathcal{L}(\mathbf{p})$ over $\mathcal{D}$.\\
\label{Lyapunov}
\end{Thm}
\begin{proof}
See Appendix A-1.
\end{proof}
If we set $B = 0$ (i.e., overload cost functions are absent) then the step-size condition in Theorem \ref{Lyapunov} is the same as that required for convergence of Algorithm 1 in \cite{GTL13}. The above results also hold if the PEV charging profile updates are not all synchronous, but follow the \textit{partially asynchronous model} as described in \cite{BT89}, in which only requires that each PEV iteratively updates its charging profile within a finite amount of time. The analysis technique and convergence results for this case are similar to that of Algorithm 2 presented in \cite{GTL13}. More discussion on the partial asynchronous implementation of our algorithm can be found in Appendix B.
\section{PEV Charging Using Primal-Dual Subgradient Method}\label{sec:primal-dual}
Due to the non-separability of the problem with respect to the decision variables (charging schedules), the dual gradient (subgradient) method (see \cite{B99} (Chapter 6)), which has been widely used to develop distributed solutions in other related contexts \cite{KMT98, LL99}, does not easily lead to a decentralized algorithm in our problem. To circumvent the issue, we apply a primal-dual subgradient method studied in \cite{NO09,Z92}.\\
The \emph{Lagrangian} of (\ref{minim}), $\mathcal{L}(\mathbf{p}, \mathbf{\mu})\text{: } \mathcal{D} \times \mathbb{R}_{+}^{L\times T} \rightarrow \mathbb{R}$, is defined as
\begin{equation}
\mathcal{L}(\mathbf{p}, \mathbf{\mu})=f\left(\mathbf{p})+\mu^{T}g(\mathbf{p}\right),
\label{Dual}
\end{equation}
where $g(\mathbf{p})=(g_{1,1}(\mathbf{p}), ..., g_{L, T}(\mathbf{p})) \leq \mathbf{0}$ is the feeder overload constraint vector, and $\mu = (\mu_{1,1}, ..., \mu_{L,T})$ is the vector of dual variables. We assume there exists a feasible solution  $\overline{\mathbf{p}}$ that satisfies all distribution feeders\rq{} constraints with a ``slack'' of $\epsilon$: $g_{l, t}(\overline{\mathbf{p}})<-\epsilon$ for $\forall l, t$, for some $\epsilon > 0$. Let $P^{max}(t)=\sum_{k=1}^{K}p_k^{max}(t)$ and define $\mu^{max}$ as in the following:
\begin{equation}
\mu^{max}=\frac{\sum_{t=1}^{T}\left\{\mathcal{V}\left(D(t)+P^{max}(t)\right)-\mathcal{V}\left(D(t)\right)\right\}}{\epsilon LT}+\frac{1}{LT}.
\label{mumax}
\end{equation}
For each feeder $l$, $1 \leq l \leq L$, define the set $\mathcal{M}_{l,t} = \{\mu_{l,t} \mid 0 \leq \mu_{l,t} \leq \mu^{max}\}$, for $t=1, \ldots, T$. The iterates of the \emph{primal-dual subgradient method} at the $m^{th}$ step, $m \geq 0$, are generated as follows \cite{NO09,Z92}:
\begin{equation}
\mathbf{p}_k^{m+1}=\mathcal{P}_{\mathcal{D}_k}\left[\mathbf{p}_k^m - \alpha \mathcal{L}_{\mathbf{p}_k}(\mathbf{p}^m,\mathbf{\mu}^m)\right], \hspace{3mm} 1 \leq k \leq K,
\label{ProjectDk}
\end{equation}
\begin{equation}
\mathbf{\mu}_{l,t}^{m+1}=\mathcal{P}_{\mathcal{M}_{l,t}}\left[\mathbf{\mu}_{l,t}^{m}+\alpha \mathcal{L}_{\mu_{l,t}}(\mathbf{p}^m, \mathbf{\mu}^m)\right], \hspace{3mm} \forall l, t,
\label{Projectpm}
\end{equation}
where $\mathcal{P}_{\mathcal{D}_k}$ and $\mathcal{P}_{\mathcal{M}_{l,t}}$ denote the projections on sets $\mathcal{D}_k$ and $\mathcal{M}_{l,t}$ respectively. The vectors $\mathcal{L}_{\mathbf{p}_k}=[\nabla \mathcal{L}(\mathbf{p}, \mathbf{\mu})]_{\mathbf{p}_k}$ and $\mathcal{L}_{\mu_{l,t}}=[\nabla \mathcal{L}(\mathbf{p}, \mathbf{\mu})]_{\mathbf{\mu}_{l,t}}$ denote the partial subgradient of $\mathcal{L}(\mathbf{p}, \mathbf{\mu})$ with respect to $\mathbf{p}_k$ and $\mu_{l,t}$, respectively. The initial vectors in $\mathbf{p}^0$  is assigned with arbitrary values as the algorithm projects the solution on the constraint plane $\mathcal{D}$ at the first step. The initial vector $\mathbf{\mu}^0$ should be feasible, e.g. $\mathbf{\mu}^{0}=0$. The scalar $\alpha > 0$ is a constant step size.

The convergence analysis that we present next follows using the convergence result on the primal-dual subgradient method in \cite{NO09}, and applying it to our context. Let $[x]^+ = \max\{x, 0\}$ and $\hat{\mathbf{p}}^m=\frac{1}{m}\sum_{i=0}^{m-1}\mathbf{p}^{i}$. As $\mathcal{D}$ is convex, thus $\hat{\mathbf{p}}^{m}\in\mathcal{D}$. Let
\begin{equation}
L_1=K\left (\sum_{t=1}^{T}\mathcal{V}\rq{}\left(D(t)+P^{max}(t)\right)+LT\mu^{max}\right ), \nonumber
\end{equation}
\begin{equation}
L_2=\sum_{l=1}^{L}\sum_{t=1}^{T}\max \left \{ \sum_{k\in \Gamma_l}p_k^{max}(t)-\eta_{l,t} P_l^{max}(t), \eta_{l,t} P_l^{max}(t)\right \}, \nonumber
\end{equation}
and $N = \max\{L_1, L_2\}$.
\begin{Thm}
As $m \rightarrow \infty$, the amount of constraint violation of the vector $\hat{\mathbf{p}}^m$, $||[g(\hat{\mathbf{p}}^{m})]^{+}||$, is upper bounded by $\frac{\alpha N^2}{2}$, and the primal cost of the vector $\hat{\mathbf{p}}^{m}$, $f(\hat{\mathbf{p}}^{m})$, is upper bounded by $f^*+\alpha N^2$ and lower bounded by $f^*-\alpha LT\mu^{max}N^2$ as $m \rightarrow \infty$.
\label{primaldual}
\end{Thm}
\begin{proof}
See Appendix A-2.
\end{proof}
\indent Theorem (\ref{primaldual}) implies that for $\alpha \rightarrow 0$, $||[g(\hat{\mathbf{p}}^{m})]^{+}||$ and $f(\hat{\mathbf{p}}^{m})$, respectively go to zero and $f^{*}$, as $m \rightarrow \infty$. In other words, the average iterates approaches the optimal solution of the primal problem in the limit, as the step size $\alpha$ becomes small. Establishing the convergence of this approach for an appropriately defined asynchronous implementation model remains open for future investigation. As compared to the approach presented in Section \ref{sec:overload-cost}, the primal-dual approach described in this section uses a fundamentally different technique in handling the constraints in (\ref{minim}). In particular, the overload cost-based approach associates a differentiable penalty function with each constraint and requires primal variable updates only; the penalty (overload cost) functions associated with the constraints get updated automatically when the primal variables are updated. In contrast, the primal-dual approach moves both the primal and dual variables in a gradient direction of the Lagrangian function (\ref{Dual}). Note that the corresponding dual function is non-differentiable (the optimal solution of the primal problem $\mathbf{P}$ is not unique). These differences have implications on the convergence properties of the two approaches, as we illustrate in Section \ref{sec:sim}.
\section{PEV Charging Profile Projection Algorithms}\label{sec:comp-algos}
\indent In this section we describe how the projection step in equations (\ref{renew1}) and (\ref{ProjectDk}), used to update the charging profile of each PEV in the two methods, can be computed efficiently. Based on the Projection Theorem defined in Proposition 2.1.3 of \cite{B99}, as $\mathcal{D}_k$ is nonempty, closed and convex subset of $\mathbb{R}^T$, there exists a unique $\mathbf{p}^{m+1}_k \in \mathcal{D}_k$, called the \emph{projection} of $\mathbf{p}_{k}^{m}-\alpha \mathbf{q}_k^m$ on $\mathcal{D}_k$, that minimizes $||\mathbf{p}_k - (\mathbf{p}_{k}^{m} - \alpha \mathbf{q}_k^m)||^2_2$ over all $\mathbf{p}_k \in \mathcal{D}_k$, where $\mathbf{q}_k^{m}=(q_k^{m}(1), ..., q_k^{m}(T))$ denote $\mathcal{L}_{\mathbf{p}_k}(\mathbf{p}^{m})$ in equation (\ref{renew1}), or $\mathcal{L}_{\mathbf{p}_k}(\mathbf{p}^{m}, \mathbf{\mu}^{m})$ in equation (\ref{ProjectDk}). Let us denote $\hat{\mathbf{p}}_k=\mathbf{p}_k^{m+1}$, then we have:
\begin{gather}
\hat{\mathbf{p}}_{k}=\arg\min_{\mathbf{p}_k \in \mathcal{D}_k} ||\mathbf{p}_k-(\mathbf{p}_{k}^{m}-\alpha \mathbf{q}_k^m)||_2^{2} \nonumber \\
= \arg\min_{\mathbf{p}_k \in \mathcal{D}_k}\sum_{t \in T_k}(p_k(t)+(\alpha q^m_k(t)-p_k^m(t))^2.
\label{renew2}
\end{gather}
The optimization problem in (\ref{renew2}) is a constrained least-square minimization problem that is also equivalent to the optimization problem defined in (8) of \cite{GTL13}. Here we propose an efficient algorithm to solve this minimization problem. Let us define $b_k(t):= \alpha q_k(t)-p_k(t)$, then (\ref{renew2}) represents a \lq\lq{}valley filling\rq\rq{} question with respect to $b_k^m(t)$ for just one vehicle, PEV $k$. Let us denote $\mathbf{b}_{k}=(b_k(1), ..., b_K(T))$, and $\mathbf{b}=(\mathbf{b}_1, ..., \mathbf{b}_K)$.
\begin{Thm}
The optimal solution, $\hat{\mathbf{p}}_k \in \mathcal{D}_k$, of equation (\ref{renew2}) is \emph{uniquely} derived as $\hat{p}_k(t) = [\min\{\lambda-b_k^m(t), p_k^{max}(t)\}]^{+}$, $\forall t \in T_k$, where $\lambda$ is a unique constant in the interval $\Lambda$:\\ $\Lambda=[\min_{t\in T_k}\left\{b_k^m(t)\right\}, \max_{t\in T_k} \left\{ b_k^m(t) +p_k^{max}(t)\right\}]$.
\label{valleyfill1}
\end{Thm}
\begin{proof}
See Appendix A-3.
\end{proof}
\indent Based on Theorem \ref{valleyfill1}, it is possible to use bisection method to find $\lambda$ in interval $\Lambda$. Let define $y(\lambda)=\sum_{t \in T_k}p_k(t)-U_k=\sum_{t \in T_k}[\min\{\lambda-b_k^m(t), p_k^{\max}(t)\}]^+ - U_k$. As $y$ is an increasing function of $\lambda$, the algorithm converges after $\log_2(|\Lambda|/{\epsilon'})$ steps, where $|\Lambda|$ is the length of interval $\Lambda$, and $\epsilon'$ is maximum approximation error. The algorithm runs in $O(T\log_2(\frac{1}{\epsilon'}))$ time, and as $\epsilon'$ gets smaller the error goes to zero.\\
\indent Based on Theorem \ref{valleyfill1}, it is also possible to compute the optimal solution of (\ref{renew2}) exactly, but in $O(T^2)$ time. This is described in the \emph{Charging Profile Projection} (CPP) algorithm.\\
\noindent \textit{Algorithm CPP}: \\
\indent Initialization: Given $b_k^m(.)$, $p_k^{max}(.)$ and $U_k$. Set $\hat{p}_k(.) \leftarrow 0$.\\
\indent Step 1: Set $a(.) = b_k^m(.)+\hat{p}_k(.)$.\\
\indent Step 2: Find $a_{min} = \min_{t\in T_k}\{a(t): a(t) < b_k^m(t)+p_k^{max}(t)\}$. If min value does not exist then Finish. Set $T_{min}=\{t: a(t)=a_{min}\}$.\\
\indent Step 3: Find $a_{next}=\min_{t \in T_k}\{a(t): a(t) > a_{min}\}$. If min value does not exist then $a_{next}=\infty$.\\
Step 4:Find $\lambda = \min\left\{\min_{t\in T_{min}}\left\{b_k^m(t)+p_k^{max}(t)\right\}, a_{next}\right\}$.\\
\indent Step 5: Let $\gamma = (\lambda - a_{min}) \times |T_{min}|$ where $|T_{min}|$ is the number of elements in $T_{min}$.\\
\indent Step 6: If $U_k > \gamma > 0$ then set $U_k \leftarrow U_k - \gamma$ and $\hat{p}_k(t)\leftarrow \hat{p}_k(t) + (\lambda - a_{min})$ for all $t \in T_{min}$ and Goto Step 1; Otherwise if $U_k \leq \gamma$ then set $\hat{p}_k(t) \leftarrow \hat{p}_k(t)+\frac{U_k}{|T_{min}|}$ for all $t \in T_{min}$ and Finish.\\
\indent In Algorithm CPP, we \lq\lq{}valley fill\rq\rq{} the different time slots $t$ in the increasing order of $b_k^m(t)$. Thus we start from the minimum level of the curve $b_k^m(t)$ in step 2, and we fill the valley up to the next minimum $b_k^m(t)$ in step 3, ensuring that the constraint for $0 \leq \hat{p}_k(t) \leq p_k^{max}(t)$ is not violated. In step 4, $\lambda$ is set to one of two cases: (i) $b_k^m(t) + p_k^{max}(t)$ for some $t \in T_k$ or (ii) $a_{next}$ which is the minimum level of available charging time slots for which charging is not assigned yet. If case (i) occurs, at least one charging time slot is omitted from the charging algorithm procedure. If case (ii) occurs, one new charging time slot gets assigned for the charging of PEV $k$. Case (i) can happen at most once for  each time slot, and therefore at most $T$ times. Similarly, case (ii) can happen at most once per time slot, and therefore at most $T$ times as well. Therefore the algorithm would finish at most in $2T$ loops. The calculations for each loop is $O(\log T)$ to find $a_{min}$, $O(\log T)$ to find  $a_{next}$, and $O(T)$ to renew $\hat{p}_k(t)$ and $U_k$. Therefore the total calculations for each loop is $O(T)$ and the total calculations for the algorithm is $O(T^2)$.

\section{Decentralized Iterative Best Response Implementation}\label{BestResponse}
\indent The two overload control methods described in this paper can be implemented as iterative best response updates by the PEVs responding to appropriately defined charging price functions determined by the utility (or aggregator). To develop this interpretation, we assume that each PEV selfishly minimizes its own individual charging cost, given a particular cost (price) function provided by the utility.\\
\indent Let the utility calculate $\mathbf{b}^{m}_k$ as follows:
\begin{equation}
\mathbf{b}^{m}_k = \alpha \mathbf{q}^{m}_k - \mathbf{p}^{m}_k,
\label{bkNearCost}
\end{equation}
where $\mathbf{q}_k^m=\mathcal{L}_{\mathbf{p}_k}(\mathbf{p}^{m})$ is derived from (\ref{cost-method-def2}) for cost-based method, and $\mathbf{q}_k^m=\mathcal{L}_{\mathbf{p}_k}(\mathbf{p}^{m}, \mathbf{\mu}^{m})$ is derived from (\ref{Dual}) for primal-dual method. The utility updates the dual variables as $\mathbf{\mu}^{m}$ as in equation (\ref{Projectpm}).\\
\indent At the $m^{th}$ iteration of both overload control methods, the utility proposes to each PEV $k$ time-dependent non-linear pricing functions $\mathbf{\Psi}_k^m(\mathbf{p}_k)=(\Psi_{k}^{m,1}(p_k(1)), ..., \Psi_{k}^{m,T}(p_k(T)))$ as follows:
\begin{gather}
\psi_k^{m,t}(p_k(t))=\mathcal{W}_k\left(b_k^{m}(t)+p_k(t)\right)-\mathcal{W}_k\left(b_k^{m}(t)\right),
\label{psikmNear}
\end{gather}
where $\mathcal{W}_k$ is a strictly convex function with continuous first and second derivatives. Note that the functions $\psi_k^{m,t}(\cdot)$ are single parameter functions ($\psi_k^{m,t}(\cdot)$ just depends on $b_k^{m}(t)$). Therefore, to communicate $\mathbf{\Psi}_k^m(\mathbf{p}_k)$ to PEV $k$, it suffices to just communicate the vector $\mathbf{b}_k^m$ to PEV $k$ at iteration $m$. Each PEV $k$ then updates its charging profile in order to minimize its charging cost:
\begin{gather}
\mathbf{p}_k^{m+1}=\arg\min_{\mathbf{p}_k \in \mathcal{D}_k}\sum_{t=1}^{T}\psi_{k}^{m,t}(p_k(t)) \nonumber \\
=\arg\min_{\mathbf{p}_k \in \mathcal{D}_k}\sum_{t \in T_k}\left(b_k^{m}(t)+p_k(t)\right)^2.
\label{PEVcost}
\end{gather}
Using equation (\ref{bkNearCost}), we observe that the RHS of (\ref{PEVcost}) is the same as the RHS of (\ref{renew2}), and can thus be computed using CPP algorithm described in Section \ref{sec:comp-algos}. The PEVs then communicate the chosen schedules $\mathbf{p}_k^{m+1}$ back to the utility, based on which the $\mathbf{b}^{m+1}_k$ values are computed, and the process repeats until convergence.
\section{Numerical Results}\label{sec:sim}
\indent In this section we use numerical studies to evaluate the performance of the two proposed overload control approaches, in minimizing the total load variation as well as controlling the feeder overload in the distribution network. We set $\mathcal{V}(x)=x^2$ in the primal problem $\mathbf{P}$, as defined in (\ref{minim}). We use the tree-structured IEEE Bus 13 distribution network \cite{B13} with the distribution substation located at the root node of the tree; the maximum depth of tree topology, $d_{max}=6$. The hourly load demand data is obtained from \cite{AE10}, and we assume the peak total load of $5$ MW in the distribution network.

In this simulation study, we connect $50$ PEVs with homogeneous charging constraints to each node (load point) of IEEE Bus 13 residential network. Each PEV requests a charging amount of $U_k = 10$ KWh.  The maximum charging capacity of each PEV is assumed to be $1.96$ KW, and the charging start and finish times are $t_k^s=0$ and $t_k^f=T=24$, for all PEVs. The PEV total load is equal to $13 \times 50 \times 10$ KWh $= 6.5$ MWh that is $5.4$ percent of the maximum generation capacity of the IEEE Bus 13 that is evaluated as $5$ MW $\times$ 24 hours $= 120$ MWh. We assume that the capacity of each feeder (link) $l$ equals $\rho_l = \nu D_{max} r_l$, where $D_{max}=5$ MW is the maximum generation capacity of IEEE Bus 13, $r_l$ is the \textit{spot load ratio} (the base load on the link as a fraction of the total base load in the distribution network) for link $l$, and $\nu$ is the feeder overload safety factor that is set to $1.5$. The spot load ratios for each link are derived from \cite{B13}. The base load at each link $l$ is also calculated as $d_l(t)=r_l D(t)$. Therefore, we have: $P_l^{max}(t)=\rho_l-d_l(t) = r_l(\nu D_{max} - D(t))$.\\
We monitor the \textit{normalized maximum link overload} of the feeders at each time $t$, calculated as:
\begin{gather}
\max_{1 \leq l \leq L}\frac{P_l(t)-P_l^{max}(t)}{\rho_l}.
\label{NormOver}
\end{gather}
Let $s_t(\mathbf{p})=\sum_{k=1}^{K}p_k(t)$ denote the total PEV aggregate load at time $t$, and let the vector $\mathbf{s}(\mathbf{p})=(s_1(\mathbf{p}), ..., s_{T}(\mathbf{p}))$ denote the total PEV aggregate load during the time window of interest, $1, ..., T$. Let $\mathbf{p}^*$ denote the limit point to which our update algorithms converge. The normalized error at iteration $m$ is determined as the normalized 2-norm distance of the total PEV load vector at iteration $m$, to $\mathbf{p}^{*}$:
\begin{equation}
\text{Normalized error}=\frac{||\mathbf{s}(\mathbf{p}^{m})-\mathbf{s}(\mathbf{p}^*)||_2}{||\mathbf{s}(\mathbf{p}^*)||_2}. 
\end{equation}
Figures \ref{TotalLoad} and \ref{Overload} show the total load and normalized maximum overload, respectively, in the distribution network. The overload cost function used is $C_{l,t}(x)=0$ for $x<0$ and $C_{l,t}(x)=\beta_{l,t} x^{(2+\hat{\epsilon})}$ for $x \geq 0$, where $\hat{\epsilon}$ is set to $0.01$, and $\eta_{l,t}$ is set to $0.9$ for all $l,t$. We set $\beta_{l,t}$ to the minimum value, $\beta^*$, for all $l,t$, such that the maximum link overload is zero for all feeders. The simulations are done with the same update step size, $\alpha$, for both overload control methods; $\alpha$ is derived based on the bound stated in Theorem \ref{Lyapunov}. 

From Figure \ref{TotalLoad}, we observe that the variance of the total load in the primal-dual method (Section~\ref{sec:primal-dual}) and the overload cost minimization method (Section~\ref{sec:overload-cost}), are approximately the same. In the figure, note that the \textit{no overload control} case, that is proposed in \cite{GTL13}, corresponds to setting $\beta_{l,t} = 0 \ \forall l, t$, i.e., only the load variance in the distribution grid is minimized, without taking into consideration any feeder capacity constraints. The variance of the total load in the primal-dual and cost-based methods is only greater than that attained without overload control by a small factor (about $0.45$ percent).  From Fig. \ref{Overload}, we observe that both approaches are very effective in avoiding overloading of distribution feeders, as the normalized overload remains below zero or very close to it at all times. The primal-dual approach has a slightly better performance in overload control compared to the cost-based approach; the no overload control solution does significantly overload one or more of the feeders for significant periods of time.\\
\begin{figure}
\centering
\subfigure[]
{
\includegraphics[width=0.75\textwidth]{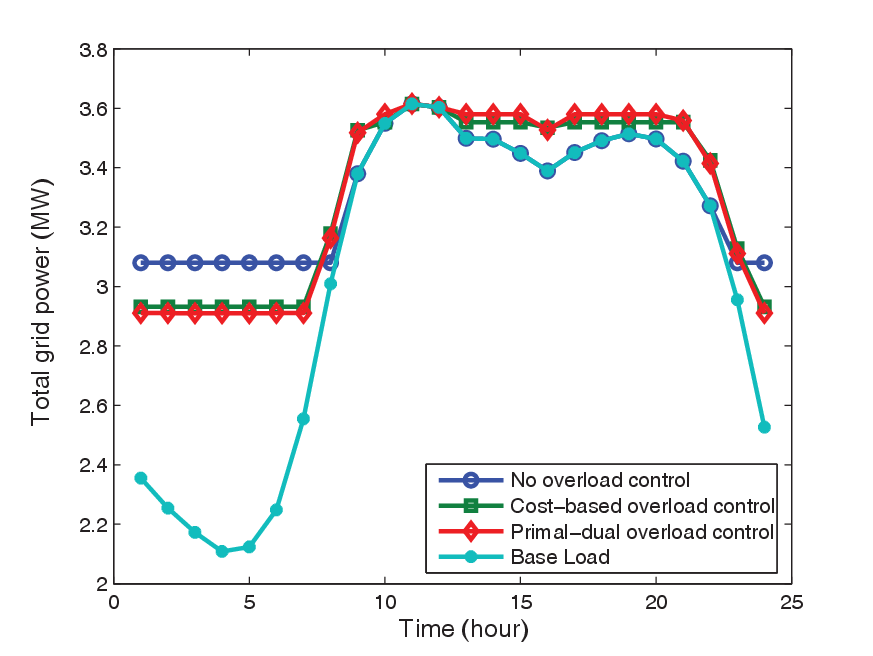}
\label{TotalLoad}
}
\subfigure[]
{
\includegraphics[width=0.75\textwidth]{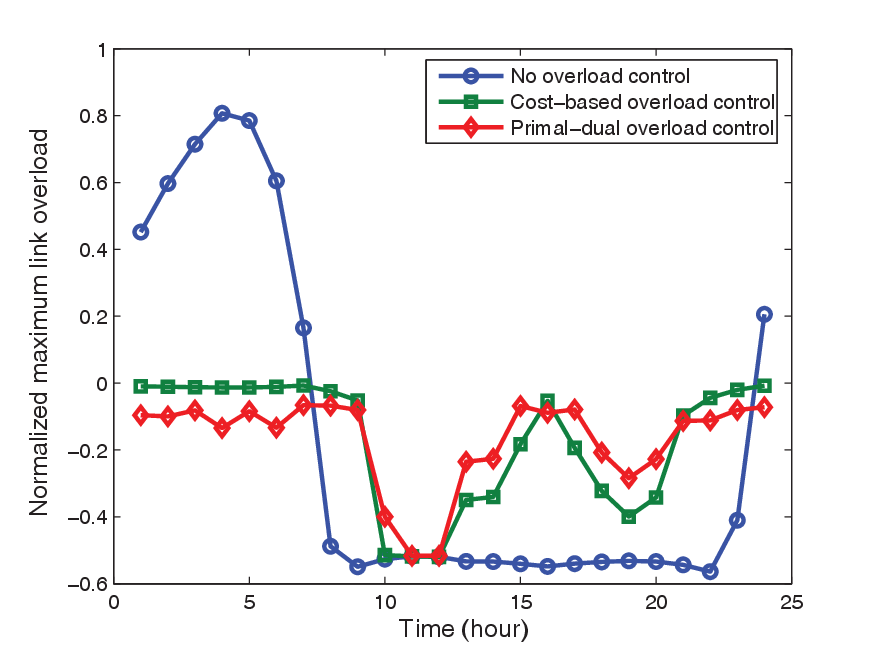}
\label{Overload}
}
\caption{(a) Total load, and (b) Normalized maximum overload of the feeders, in the IEEE Bus 13 distribution system with no overload control, and overload control using cost-minimization and primal-dual methods.}
\vspace{-5mm}
\end{figure}
\indent Figure \ref{NormError} demonstrates the normalized error vs. the number of update rounds for this simulation study. We observe that the normalized error decreases slightly faster in the cost-based approach. Moreover, the normalized error for the primal-dual approach does not monotonically decrease with increasing number of rounds, unlike the cost-based approach.
\begin{figure}
\centering
\includegraphics[width=0.75\textwidth]{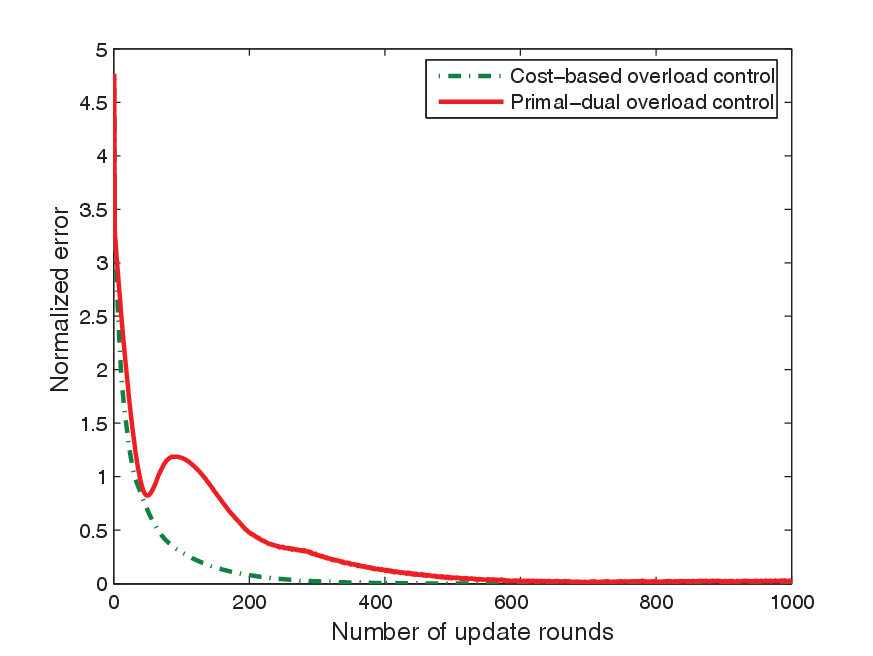}
\caption{Normalized error vs. number of update rounds for the two methods.}
\label{NormError}
\end{figure}
\section{Conclusion} \label{sec:concl}
\indent In this paper we addressed the question of overload control on the feeders of a distribution network that is being used in charging PEVs. We proposed two different (sub-) gradient optimization methods for the overload control problem -- one based on cost (penalty) minimization and another based on primal-dual approach -- both of which are amenable to distributed implementation through back-and-forth communication between the utility (aggregator) and the PEVs. We showed theoretically that the iterative algorithms based on both approaches converge to attain near-optimal load variance while ensuring that the feeders are not overloaded. 
The cost-based approach requires that the step size be smaller than a certain system-dependent threshold, to ensure convergence. The primal-dual approach does not require a specific upper bound on the step size, although the degree of approximation depends linearly on the step-size in addition to other system parameters. In simulations, we observed that while both overload control methods are quite effective in controlling feeder overload. The primal-dual approach seems slightly better in terms of maximum feeder overload control, while the cost-based approach attains faster and smoother (monotonic) convergence.

While we discuss a best-response implementation of our algorithms in Section \ref{BestResponse}, a formal game-theoretic treatment of this problem is beyond the scope of the current work. In this respect, the approach that we propose differs from that in \cite{MCH10}, \cite{GKBG13}, where the PEV charging problem (without feeder overload constraints) has been analyzed as a strategic game of complete information between PEVs, and the Nash equilibrium properties of the game has been studied. Although it seems likely that the two approaches that we propose in this paper could be utilized (for incorporating feeder overload constraints) in such game-theoretic frameworks as well, that question remains open for further investigation.

\bibliographystyle{ieeetran}
\bibliography{ref}

\appendix
\renewcommand{\thesection}{\arabic{section}}
\subsection{Proofs}
\subsubsection{Proof of Theorem \ref{Lyapunov}}\label{A1}
\indent For the charging profile $\mathbf{p}=(p_1(1), p_1(2), ..., p_K(T))$, $\nabla [\mathcal{L}(\mathbf{p})]_{k,t}=\frac{\partial \mathcal{L}(\mathbf{p})}{\partial p_k(t)}$, the $\{(k-1)T+t\}^{th}$ component of $\nabla \mathcal{L}(\mathbf{p})$, $1 \leq k \leq K$, $1 \leq t \leq T$, is calculated as in the following:\\
\begin{equation}
[\nabla \mathcal{L}(\mathbf{p}]_{k,t}= \mathcal{V}\rq{} \left (D(t)+P(t)\right)+\sum_{l:k \in l}C_{l,t}\rq{}\left(P_l(t)-\eta_{l,t} P_l^{max}(t)\right)
\label{nablaL}
\end{equation}
\indent Let $\mathbf{p}, \mathbf{r} \in \mathcal{D}$ are two different feasible charging profiles and let define $R(t)=\sum_{k=1}^{K}r_k(t)$ and $R_l(t)=\sum_{k \in \Gamma_l}r_k(t)$. Then have:
\begin{gather}
||\nabla \mathcal{L}(\mathbf{p}) - \nabla \mathcal{L}(\mathbf{r})||_2^2 = \nonumber \\
\sum_{t=1}^{T}\sum_{k=1}^{K} \Bigg{\{} \left( \mathcal{V}\rq{} \left( D(t)+P(t)\right) - \mathcal{V}\rq{} \left( D(t)+R(t)\right)\right)^2 \nonumber \\
+ 2 \left(\mathcal{V}\rq{}(D(t)+P(t))-\mathcal{V}\rq{}(D(t)-R(t))\right) \nonumber \\
\times \sum_{l \in \Pi_k}\left(C_l\rq{}(P_l(t)-\eta_{l,t}P_l^{max}(t))-C_l\rq{}(R_l(t)-\eta_{l,t}P_l^{max}(t))\right)\nonumber \\
+\left(\sum_{l \in \Pi_k}C_l\rq{}(P_l(t)-\eta_{l,t}P_l^{max}(t))-C_l\rq{}(R_l(t)-\eta_{l,t}P_l^{max}(t))\right)^2\Bigg{\}}.
\label{deltaL}
\end{gather}
\indent As the first and second derivative of $\mathcal{V}(x)$ and $C_{l,t}(x)$, for all $l, t$, is continuous on $\mathbb{R}$, then based on the mean value theorem there exists some $\tilde{x}(t)$ between $D(t)+P(t)$ and $D(t)+R(t)$, and some $\tilde{x}_l(t)$ between $P_l(t)-\eta_{l,t}P_l^{max}(t)$ and $R_l(t)-\eta_{l,t}P_l^{max}(t)$ for all $l, t$ such that:\\
\begin{equation}
\mathcal{V}\rq{}\rq{}(\tilde{x}(t))=\frac{\mathcal{V}\rq{}(D(t)+P(t))-\mathcal{V}\rq{}(D(t)+R(t))}{P(t)-R(t)}.\\
\label{MeanValueV}
\end{equation}
and
\begin{equation}
C_{l,t}^{\rq{}\rq{}}(\tilde{x}_{l}(t))=\frac{C_{l,t}^{\rq{}}(P_l(t)-\eta_{l,t}P_l^{max}(t))-C_{l,t}^{\rq{}}(R_l(t)-\eta_{l,t}P_l^{max}(t))}{P_l(t)-R_l(t)}.\\
\label{MeanValue}
\end{equation}
Substituting (\ref{MeanValueV}) and (\ref{MeanValue}) in (\ref{deltaL}) and using the assumption that $\mathcal{V}\rq{}\rq{}(x) \leq B_1$, for all $x \in [D(t), D(t)+\max_{t}\left\{\sum_{k=1}^{K}p_{k}^{max}(t)\right\}]$, and $C_{l,t}^{\rq{}\rq{}}(x) \leq B$, for all $x \in [-\max_{l,t}\eta_{l,t}P_{l}^{max}(t),$\\$ \max_{l,t}\{\sum_{k \in \Gamma_l}p_k^{max}(t)-\eta_{l,t}P_{l}^{max}(t)\}]$, we have:\\
\begin{gather}
||\nabla \mathcal{L}(\mathbf{p}) - \nabla \mathcal{L}(\mathbf{r})||_2^2 \leq \sum_{t=1}^{T}\Bigg{\{}B_{1}^{2}K(P(t)-R(t))^2 \nonumber \\
+2B_{1}B(P(t)-R(t)) \sum_{k=1}^{K}\sum_{l \in \Pi_k}(P_l(t)-R_l(t)) \nonumber \\
+B^2\sum_{k=1}^{K}(\sum_{l \in \Pi_k}P_l(t)-R_l(t))^2\Bigg{\}}.
\label{deltaL1}
\end{gather}
Rearranging the terms and using the inequality $l_k \leq d_{max}$ have:\\
\begin{gather}
\sum_{k=1}^{K}\sum_{l \in \Pi_k}\left(P_l(t)-R_l(t)\right)=\sum_{k=1}^{K}(\sum_{l \in \Pi_k}K_l)(p_k(t)-r_k(t)) \nonumber \\
\leq d_{max}(P(t)-R(t))
\label{Rearrange1}
\end{gather}
Using Jensen\rq{}s inequality twice and from the fact that $|l_k| \leq d_{max}$ and $K_l \leq K$ with rearranging the terms have:\\
\begin{gather}
\sum_{k=1}^{K}(\sum_{l \in \Pi_k}P_l(t)-R_l(t))^2 \leq \sum_{k=1}^{K}\sum_{l \in \Pi_k} l_{k}(P_l(t)-R_l(t))^2 \nonumber \\
\leq d_{max} \sum_{k=1}^{K}\sum_{l \in \Pi_k}(P_l(t)-R_l(t))^2 \nonumber \\
= d_{max} \sum_{k=1}^{K}\sum_{l \in \Pi_k}(\sum_{j \in \Gamma_l}p_j(t)-r_j(t))^2 \nonumber \\
\leq d_{max} \sum_{k=1}^{K}\sum_{l \in \Pi_k} K_l \sum_{j \in \Gamma_l}(p_j(t)-r_j(t))^2 \nonumber \\
= d_{max} \sum_{k=1}^{K}(\sum_{l \in \Pi_k}K_l^2)(p_k(t)-r_k(t))^2 \leq d_{max}^2K^2\sum_{k=1}^{K}(p_k(t)-r_k(t))^2.
\label{Rearrange2}
\end{gather}
Using the Jensen\rq{}s inequality:
\begin{equation}
\sum_{t=1}^{T}(P(t)-R(t))^2 \leq K||\mathbf{p}-\mathbf{r}||_2^2,
\label{Jensen1}
\end{equation}
and substituting (\ref{Rearrange1}) and (\ref{Rearrange2}) in (\ref{deltaL1}) results in:\\
\begin{equation}
||\nabla \mathcal{L}(\mathbf{p})- \nabla \mathcal{L}(\mathbf{r})||_2^2 \leq (K(B_1+d_{max}B))^2||\mathbf{p}-\mathbf{r}||_2^2.
\label{squaredel}
\end{equation}
\indent Based on the proof of proposition 2.3.2 and equation (2.30) in \cite{B99}, $\mathcal{L}$ is a decreasing function of $\mathbf{p}^m_k$:\\
\begin{equation}
\mathcal{L}(\mathbf{p}^{m+1})-\mathcal{L}(\mathbf{p}^{m}) \leq (\frac{K}{2}(B_1+d_{max}B)-\frac{1}{\alpha})||\mathbf{p}^{m+1}-\mathbf{p}^{m}||_2^2
\label{Ldec}
\end{equation}
 for $0 < \alpha < (K(B_{1}+d_{max}B))^{-1}$.\\
 \indent As the right side of equation (\ref{Ldec}) is non-positive, $\mathcal{L}(\mathbf{p}^{m})=\mathcal{L}(\mathbf{p}^{m+1})$ if and only if $\mathbf{p}^{m}=\mathbf{p}^{m+1}$ for some $m \geq 1$. In this case, $\mathbf{\tilde{p}}^{*}=\mathbf{p}^{m}$ is a limit point of $\{\mathbf{p}^{m}\}_{m=1}^{\infty}$ as $\mathbf{p}^{n}=\mathbf{\tilde{p}}^{*}$ for all $n \geq m$ and the sequence $\{\mathcal{L}(\mathbf{p^{m}})\}_{m=1}^{\infty}$ converges decreasingly to $\mathcal{L}^*=\mathcal{L}(\mathbf{\tilde{p}}^*)$.\\
 \indent For the only other possible case, $\{\mathcal{L}(\mathbf{p}^{m})\}_{m=1}^{\infty}$ is a strictly decreasing sequence of $\mathbf{p}^m$, i.e. $\mathcal{L}(\mathbf{p}^{m}) < \mathcal{L}(\mathbf{p}^{n})$ for all $n > m$ and thus $\mathbf{p}^m \neq \mathbf{p}^n$ for all $n \neq m$. As $C_{l,t}(x)$, for all $1 \leq l \leq L$, is a non-negative and continious function, therefore $\mathcal{L}(\mathbf{p})$ is also a non-negative and continious function and bounded below by $\sum_{t=1}^{T}\mathcal{V}(D(t))$. Thus based on the monotone convergence theorem and continuity of $\mathcal{L}(\mathbf{p})$, $\{ \mathcal{L}(\mathbf{p}^{m}) \}_{m=1}^{\infty}$ converges to a limit point $\mathcal{L}^{*}=\lim_{m \rightarrow \infty} \mathcal{L}(\mathbf{p}^{m})$. As $\mathcal{D}$ is compact, based on the Bolzano-Weierstrass theorem, the infinite subset $\{\mathbf{p}^{m}\}_{m=1}^{\infty} \subseteq \mathcal{D}$ has a limit point $\mathbf{\tilde{p}}^{*} \in \mathcal{D}$, i.e. there exists a subsequence $\{\mathbf{\tilde{p}}^{m}\}_{m=1}^{\infty} \subseteq \{ \mathbf{p}^{m}\}_{m=1}^{\infty}$ such that $\lim_{m \rightarrow \infty} \mathbf{\tilde{p}} = \mathbf{\tilde{p}}^{*}$. As $\mathbf{p}^m \neq \mathbf{p}^n$ for all $m \neq n$ therefore $\{\mathbf{\tilde{p}}^{m}\}_{m=1}^{\infty} \subseteq \{ \mathbf{p}^{m}\}_{m=1}^{\infty}$ is also an infinite set. From the continiuity of $\mathcal{L}$ over $\mathcal{D}$ have: $\mathcal{L}^{*}=\lim_{m \rightarrow \infty} \mathcal{L}(\mathbf{p}^{m})=\lim_{m \rightarrow \infty} \mathcal{L} (\mathbf{\tilde{p}}^{m}) = \mathcal{L}(\mathbf{\tilde{p}}^*)$.  \\
 \indent Based on proposition 2.3.2 in \cite{B99}, in both cases, the limit point $\mathbf{\tilde{p}}^{*}$ is stationary, i.e. it satisfies the optimality condition (2.1) in \cite{B99}:\\
\begin{equation}
\nabla \mathcal{L}(\mathbf{p}^*)^{T}(\mathbf{p}-\mathbf{p}^*) \geq 0, \hspace{5mm} \forall \mathbf{p}\in \mathcal{D}.
\label{stationary}
\end{equation}
\indent As $C_{l,t}(x)$ is a convex function, $\mathcal{L}(\mathbf{p})$ is also convex over $\mathcal{D}$, and based on proposition 2.1.2 in \cite{B99}, the limit point $\mathbf{\tilde{p}}^{*}$ minimizes $\mathcal{L}$ over $\mathcal{D}$: $\mathcal{L}^*=\mathcal{L}(\mathbf{\tilde{p}^*}) \leq \mathcal{L}(\mathbf{p})$ for $\forall \mathbf{p} \in \mathcal{D}$. This concludes the proof of Theorem 1.\\
 \indent \emph{Partially asynchronous updates:} Note that in equation (\ref{squaredel}), the vectors $\mathbf{p}$ and $\mathbf{r}$ are arbitrary vectors in $\mathcal{D}$. Furthermore, based on the cartesian product property of $\mathcal{D}=\mathcal{D}_1\times ...\times \mathcal{D}_{K}$, and using the characteristic property of the projection (Proposition 2.1.3 in \cite{B99}) on each feasible set, $\mathcal{D}_k$, all the equations in the proof of Proposition 2.3.2 in \cite{B99} also hold if not all of PEVs update their charging profiles in an iteration. Therefore, equation (\ref{Ldec}) holds for partially synchronous-based iterative update method that is described in chapter 7 in \cite{BT89}. In order to achieve convergence, we assume that each PEV updates its charging profile at least once during each finite length cycle of iterations, where a cycle refers to a set of successive iterations that each PEV updates its charging profile at least once. Then a convergence result similar to Theorem 1 can be shown to hold in the partially asynchronous update model as well. \\
\subsubsection{Proof of Theorem \ref{primaldual}}\label{A2}
\indent From equation (\ref{cost-gradient}):
\begin{equation}
\mathcal{L}_{\mathbf{p}_k}(\mathbf{p}^m) = \sum_{t=1}^{T} \Bigg{\{}\mathcal{V}'\left(D(t)+\sum_{k=1}^{K}p^m_k(t)\right) + \sum_{l \in \Pi_k} C'_{l,t}\left(\sum_{k \in \Gamma_l}p^m_k(t)-\eta_{l,t} P_l^{max}(t)\right)\Bigg{\}}. \label{cost-gradient}
\end{equation}
 and from the fact that $||\mathbf{x}|| \leq ||\mathbf{x}||_1$ for any vector $\mathbf{x}$, we have: $||\mathcal{L}_{\mathbf{p}}(\mathbf{p}^{m}, \mathbf{\mu}^{m})|| \leq L_1$. From equation (\ref{dual-gradient-mu}):
\begin{equation}
\mathcal{L}_{\mu_{l,t}}(\mathbf{p}^m, \mathbf{\mu}^m) \ = \ g_{l,t}(\mathbf{p}^m(t)) \ = \ P^m_l(t) - \eta_{l,t}P_l^{max}(t),
\label{dual-gradient-mu}
\end{equation}
we have: $||\mathcal{L}_{\mathbf{\mu}}(\mathbf{p}^{m}, \mathbf{\mu}^{m})) \leq L_2$ for all $m$. Therefore, the subgradients $\mathcal{L}_{\mathbf{p}}(\mathbf{p}^{m}, \mathbf{\mu}^{m})$ and $\mathcal{L}_{\mathbf{\mu}}(\mathbf{p}^{m}, \mathbf{\mu}^{m})$ are uniformly bounded by $N=\max\{L_1, L_2\}$.\\
\indent Let $q: \mathbb{R}_{+}^{L} \rightarrow \mathbb{R}$ is the dual objective function: $\mathbf{q}({\mu})=\inf_{\mathbf{p} \in \mathcal{D}}\mathcal{L}(\mathbf{p}, \mathbf{\mu})$. We define $q_0 = q(0) = \inf_{\mathbf{p} \in \mathcal{D}}f(\mathbf{p})$. The vector $\overline{\mathbf{p}} \in \mathcal{D}$ defined in section IV with $g_{l,t}(\overline{\mathbf{p}}) < \epsilon$, for some $\epsilon > 0$, satisfies the slater condition described in assumption (2) in \cite{NO09}.\\
\indent Using the inequalities: $f(\overline{\mathbf{p}}) \leq \sum_{t=1}^{T}\mathcal{V}(D(t)+P^{max}(t))$, $q_0 \geq \sum_{t=1}^{T}\mathcal{V}(D(t))$ and $\gamma = \min_{\forall l, t}\{-g_{l,t}\} > \epsilon$, have:\\
\begin{gather}
\frac{f(\overline{\mathbf{p}})-q_0}{\gamma}<\frac{\sum_{t=1}^{T}\mathcal{V}(D(t)+P^{max}(t))-\sum_{t=1}^{T}\mathcal{V}(D(t))}{\epsilon}=LT\mu^{max}-1.
\label{mumax-ineq}
\end{gather}
Let define the scalar $r>1$ as $r=LT\mu^{max}-\frac{f(\overline{\mathbf{p}})-q_0}{\gamma}$ and the set $\mathcal{M}=\{ \mathbf{\mu} \mid \mu_{l, t} \in \mathcal{M}_{l,t}, \forall l, t\}$, have: $||\mathbf{\mu}|| \leq ||\mu||_1=\sum_{l=1}^L \sum_{t=1}^{T} \mu_{l,t} \leq LT\mu^{max}=\frac{f(\overline{\mathbf{p}})-q_0}{\gamma}+r$, for all $\mathbf{\mu} \in \mathcal{M}$.\\
\indent Let $\mathbf{p}^{*} \in \mathcal{D}$ be an optimal solution. As the slater condition and bounded subgradient assumptions hold, based on the part (a) of the proposition 2 in \cite{NO09} and $r>1$, $||g(\hat{\mathbf{p}}^{m})^{+}||$ is upper bounded by $\frac{2}{m \alpha}\left(\mu^{max}\right)^2+\frac{||\mathbf{p}^{0}-\mathbf{p}^{*}||}{2m\alpha}+\frac{\alpha N^2}{2}$ and as $m \rightarrow \infty$, the upper bound goes to $\frac{\alpha N^2}{2}$.\\
\indent Furthermore based on the part (b) and (c) of the proposition 2 in \cite{NO09}, $f(\hat{\mathbf{p}}^{m})$ is upper bounded by $f^{*}+\frac{||\mu^{0}||^2}{2m\alpha}+\frac{||\mathbf{p}^{0}-\mathbf{p}^{*}||^2}{2m\alpha}+\alpha N^2$ and is lower bounded by $f^{*}-LT\mu^{max}||g(\tilde{\mathbf{p}}^{m})^{+}||$. As $m \rightarrow \infty$, the upper bound and lower bound go to $f^{*}+\alpha N^2$ and $f^{*} - LT\mu^{max} \alpha N^2$, respectively.
\subsubsection{Proof of Theorem \ref{valleyfill1}} \label{A3}
\indent Let denote $\mathcal{X}_k=\{\mathbf{p}_k| 0 \leq p_k(t) \leq p_k^{max}(t), \forall t \in T_k\}$. Let define $h(\mathbf{p}_k)=\sum_{t\in T_k}(b_k(t)+p_k(t))^2$ and let us define the Lagrangian function:
\begin{equation}
\mathcal{L}(\mathbf{p}_k, \lambda)=\sum_{t\in T_k}(p_k(t)+b_k(t))^2-2\lambda(\sum_{t\in T_k}p_k(t)-U_k).
\label{LagrangeThm3}
\end{equation}
\indent The objective function in equation (\ref{renew2}), $h(\mathbf{p}_k)$, is convex and continuously differentiable in $\mathcal{X}_k$ and $\mathcal{X}_k$ is a bounded polyhedral set that contains all the inequality constraints for $\mathbf{p}_k$. Therefore, based on Proposition 3.4.1 in \cite{B99}, the feasible charging profile, $\mathbf{p}^{*}_k$, is an optimal solution for equation (\ref{renew2}), if and only if there exists scalar $\lambda^* \in \mathbb{R}$ such that $\mathbf{p}_k^{*}$ minimizes $\mathcal{L}(\mathbf{p}_k, \lambda^*)$ over $\mathcal{X}_k$:
\begin{eqnarray}
\mathbf{p}_k^{*}=\arg\min_{\mathbf{p}_k \in \mathcal{X}_k}\mathcal{L}(\mathbf{p}_k, \lambda^{*})=\arg\min_{\mathbf{p}_k\in\mathcal{X}_k}\sum_{t=t_k^s}^{t_k^f} (p_k(t)+b_k(t) - \lambda^*)^2 \nonumber\\
=\sum_{t=t_k^s}^{t_k^f}\arg\min_{0 \leq p_k(t) \leq p_k^{max}(t)}(p_k(t)+b_k(t) - \lambda^*)^2.
\label{LagrangeDual}
\end{eqnarray}
\indent The optimal solution for $p_k^{*}(t)$, $t \in T_k$, that minimizes each term of $(p_k(t)+b_k(t) - \lambda^*)^2$ in equation (\ref{LagrangeDual}), is derived based on the only three possible cases as in the following:
\begin{eqnarray}
p^{*}_k(t)=
\begin{cases}
0,  \hspace{5mm} \lambda^* \leq b_k(t),\\
\lambda^* - b_k(t), \hspace{5mm} b_k(t) < \lambda^* < b_k(t) + p_k^{max}(t), \hspace{10mm} t \in T_k\\
p_k^{max}(t),   \hspace{5mm} b_k(t) + p_k^{max}(t) \leq \lambda^*.
\end{cases}
\label{lambdaThm3}
\end{eqnarray}
\indent Based on equation (\ref{lambdaThm3}),we have: $p_k^{*}(t)=[\min\left\{\lambda^*-b_k(t), p_k^{max}(t)\right\}]^{+}$, $\forall t \in T_k$. As the function $y(\lambda^*) = \sum_{t=t_k^s}^{t_k^f} [\min\left\{\lambda^*-b_k(t), p_k^{max}(t)\right\}]^{+}-U_k$ is a strictly increasing function of $\lambda^* \in \Lambda$ and it ranges in $-U_k \leq y(\lambda^*) \leq \sum_{t \in T_k}p_k^{max}(t)-U_k$ for $\lambda^* \in \Lambda$, therefore there exists a unique $\lambda^* \in \Lambda$ such that $y(\lambda^*)=0$ to satisfy the constraint in equation (\ref{Uk}). 
Thus the optimal solution of equation (\ref{renew2}) is uniquely derived as $p_k^{*}(t)=[\min_{t \in T_k}\{\lambda^{*}-b_k(t)\}, p_k^{max}(t)]^{+}$ for a unique $\lambda^{*} \in \Lambda$.\\
\subsection{Additional Simulation Results}
\subsubsection{Effects of scaling factors in the cost-based overload control approach}
In this section we investigate the effects of the scaling factors, $\beta_{l,t}$, and overload control factors $\eta_{l,t}$, for all $l,t$, in the performance of the cost-based overload control method through the simulations. In our simulations, we consider $\beta_{l,t}$ and $\eta_{l,t}$ as constant values for all $l,t$, denoted as $\beta$ and $c$, respectively. Let $\beta^*$ denote the minimum $\beta$ such that the maximum normalized link overload derived as in equation (\ref{NormOver}) is not larger than zero for all $t$.\\
\indent Figures \ref{Betac90} and \ref{Betac95} demonstrate the total load and maximum normalized link overload of the grid using cost-based overload control approach varying $\beta$ and $c=0.9$ and $c=0.95$, respectively. As it is seen in these figures, increasing $\beta$ decreases maximum normalized link overload. The reduction is noticeable when varying $\beta$ from $\beta^*$ to $10\beta^*$, but as $\beta$ gets larger, increasing $\beta$ does not have any significant effect on minimizing overload and total load variance.\\
\indent Note that the second derivative of the cost function $C_{l,t}^{''}(x)=(2+\hat{\epsilon})(1+\hat{\epsilon})\beta$ is bounded by $B=2.04\beta$, for $\hat{\epsilon}=0.01$. Based on Theorem \ref{Lyapunov}, the upper bound of the step size $\alpha$ required for convergence has an inverse relationship with $B$, and therefore with $\beta$.  Therefore as $\beta$ increases, the step size must be made smaller, which reduces the speed of convergence.
\begin{figure}[H]
\centering
\subfigure[]
{
\includegraphics[width=0.75\textwidth]{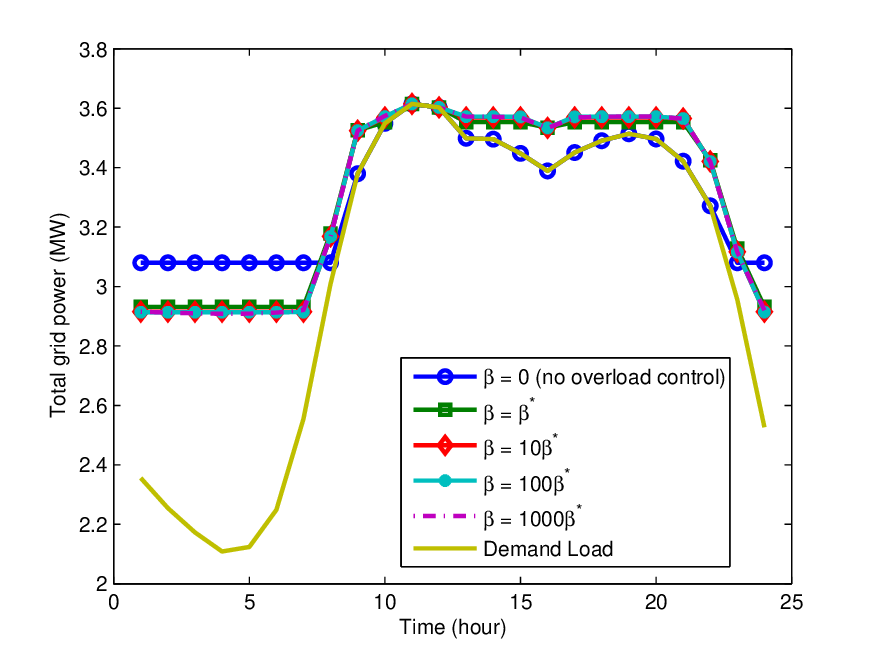}
\label{TotalloadBetac90}
}
\subfigure[]
{
\includegraphics[width=0.75\textwidth]{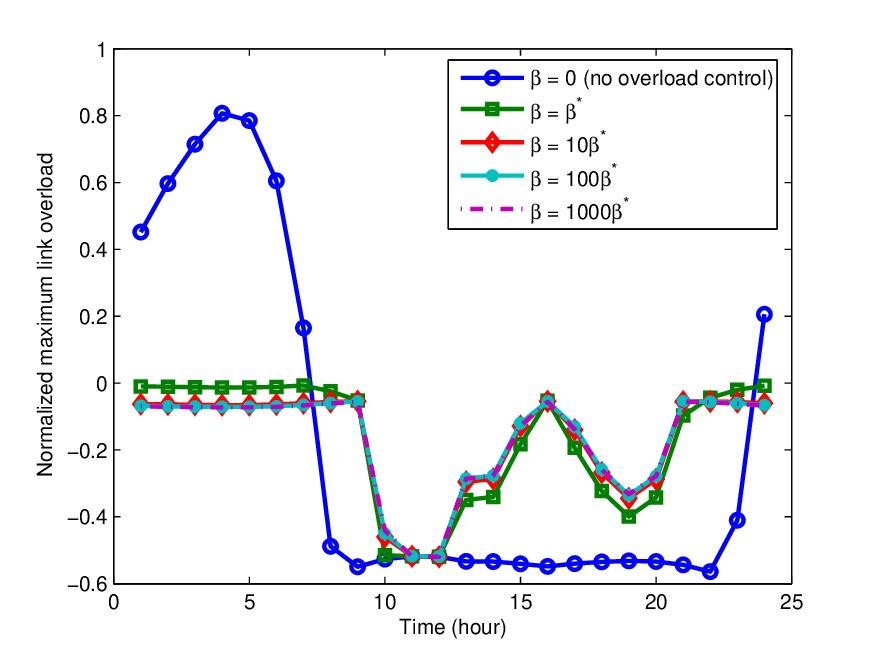}
\label{OverloadBetac90}
}
\caption{(a) Total load, and (b) Normalized maximum overload of the feeders in IEEE Bus 13 distribution system using cost-based overload control with $c=0.9$ and varying $\beta$.}
\label{Betac90}
\end{figure}
\begin{figure}[H]
\centering
\subfigure[]
{
\includegraphics[width=0.75\textwidth]{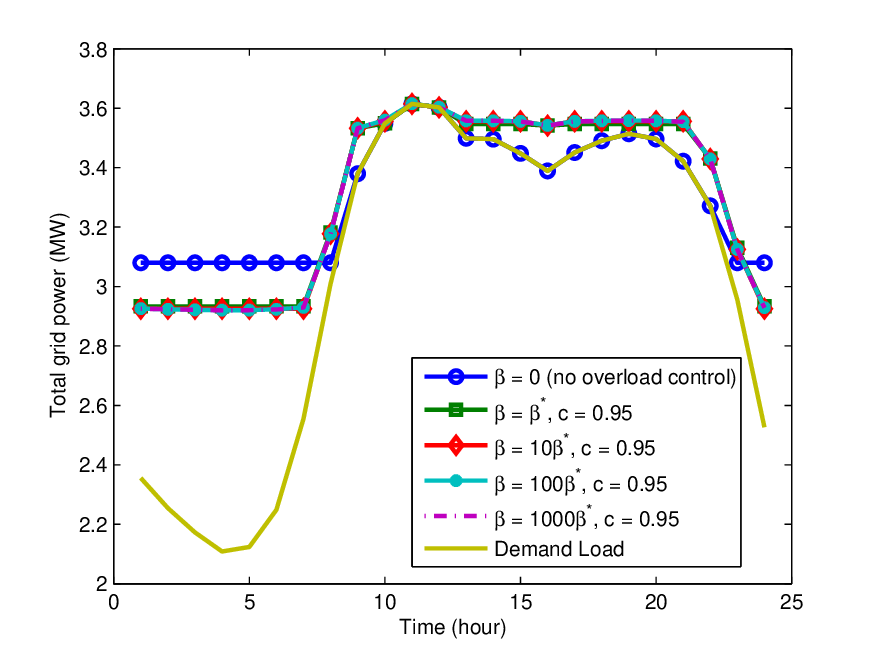}
\label{TotalloadBetac95}
}
\subfigure[]
{
\includegraphics[width=0.75\textwidth]{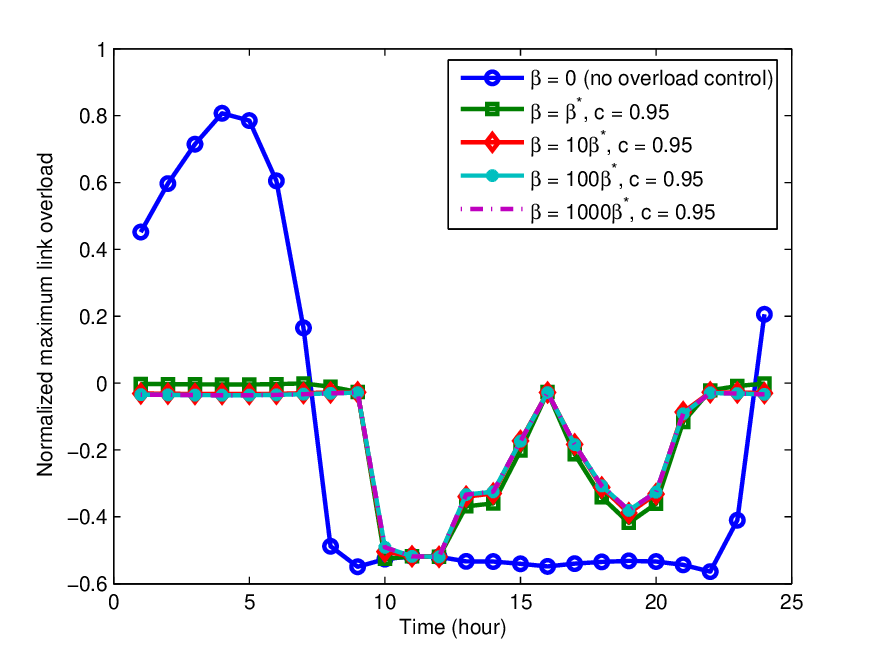}
\label{OverloadBetac95}
}
\caption{(a) Total load, and (b) Normalized maximum overload of the feeders in IEEE Bus 13 distribution system using cost-based overload control with $c=0.95$ and varying $\beta$.}
\label{Betac95}
\end{figure}
\subsubsection{Asynchronous updates}
\indent Based on the discussion towards the end of Appendix A-1 the charging profiles of PEVs can be updated according to the partially asynchronous model, in the cost-based overload control method. In this case, it is not necessarily that all the PEVs update their charging profile in each iteration, but it is necessary that each PEV updates its charging profile in finite time. The asynchronous update model (policy) we implement next makes even weaker assumptions on the synchrony of the updates: a random PEV updates its charging profile at each iteration. In primal-dual asynchronous-based update, the dual variables are also updated at the end of each iteration. 
Figure \ref{Asynch} compares the normalized error vs. average number of updates per PEV, for the cost-based and primal-dual overload control methods with asynchronous (random) updates.  
Comparing with Figure \ref{NormError}, we observe that the convergence rate of the asynchronous updates (random update order), the convergence rate is comparable to the synchronous charging profile update policy.\\
\indent Although the simulation results show that the asynchronous-based charging profile update in the primal-dual approach converges to the optimal charging profile, the theoretical proof of the convergence is remained under further investigation in our future work.
\begin{figure}[H]
\centering
\includegraphics[width=0.75\textwidth]{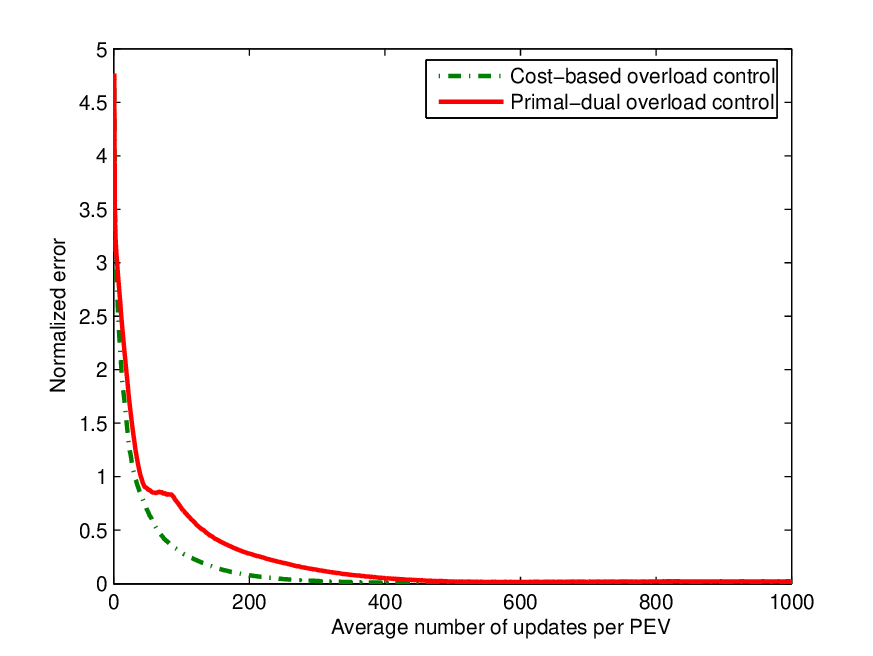}
\caption{Normalized error vs. Average number of updates per PEV for the asynchronous charging profile updates in cost-based and primal-dual overload control methods.}
\label{Asynch}
\end{figure}
\end{document}